\documentclass{article}
\usepackage{graphicx} 
\usepackage[utf8]{inputenc}
\usepackage{amsmath, amssymb, amsthm}
\usepackage{geometry}
\usepackage{authblk}
\usepackage{tikz}
\usetikzlibrary{positioning, shapes.geometric, arrows.meta, calc}
\usepackage{hyperref}

\theoremstyle{plain}
\newtheorem{theorem}{Theorem}[section]
\newtheorem{lemma}[theorem]{Lemma}

\newtheorem{proposition}[theorem]{Proposition}
\newtheorem{corollary}[theorem]{Corollary}
\newtheorem{innerclaim}{Claim}
\newenvironment{claim}
  {\begin{innerclaim}}
  {\end{innerclaim}}
\usepackage{enumitem}

\theoremstyle{definition}
\newtheorem{definition}{Definition}[section]

\newtheorem{example}{Example}[section]

\theoremstyle{remark}

\usepackage{todonotes}
\setuptodonotes{inline}

\title{The Minimum-Weight Mixed Dominating Set on Threshold Graphs}
\author[1]{Emiliano Lancini%
  \thanks{\href{mailto:emiliano.lancini@lamsade.dauphine.fr}
  {\nolinkurl{emiliano.lancini@lamsade.dauphine.fr}}}}

\author[2]{Oulin Yang%
  \thanks{\href{mailto:oulin_yang@outlook.com}
  {\nolinkurl{oulin_yang@outlook.com}}}}

\affil[1]{LAMSADE, CNRS UMR 7243,
Université Paris Dauphine--PSL, Paris, France}

\affil[2]{Institut Polytechnique de Paris,
Palaiseau, France}

\date{August 2026}
\begin{document}
	
	\maketitle
	
	\begin{abstract}
		We study the minimum-weight mixed dominating set problem on threshold graphs. In this problem, vertices and edges have weights, and the goal is to find a mixed set of minimum total weight that dominates every vertex and edge of the graph. We first show that arbitrary weights can be reduced to non-negative weights without changing the asymptotic running time. By adapting a reduction to the minimum-weight edge cover given in Ferrarini, Kober, Lancini, and Yuditsky~\cite{ferrariniKoberLanciniYuditsky}, we obtain an $\mathcal{O}(n^5)$-time algorithm for the minimum weight mixed dominating set problem on threshold graphs.
	\end{abstract}
	
	\section{Introduction}
	
	
	The domination problem in graph theory asks for a minimum-size vertex subset such that every vertex outside the subset has a neighbor in it. In this paper, we study the \emph{mixed dominating set problem}. A mixed dominating set is a subset of the vertices and edges of a graph that dominates every vertex and edge. A vertex dominates itself, its neighbors, and its incident edges; an edge dominates itself, its endpoints, and every edge sharing an endpoint with it.  In particular, we are interested in the \emph{Minimum Weight Mixed Dominating Set Problem} (MWMDS), where each vertex and edge is assigned a weight, and the objective is to find the mixed dominating set of minimum the total weight.
    
	Domination problems find application in various fields. For instance, mixed domination can be applied to monitoring of sensible areas and energy and telecommunication networks.
	
	For example, the weighted problem has an intuitive application in urban security and patrol planning. Consider a city road network modeled as a graph, where intersections are vertices and streets are edges. Fixed guard booths can be placed at selected intersections, while patrol cars can be assigned to selected streets. Vertex and edge weights represent the corresponding deployment costs. A minimum-weight mixed dominating set then gives a lowest-cost plan in which every intersection and street is monitored by at least one selected booth or patrol route.
	
	\subsection{Related Work}
	
	The Mixed Dominating Set problem, also known as the total cover problem, was introduced by Alavi, Behzad, Lesniak-Foster, and Nordhaus~\cite{alavi1977}. It was later studied in several works on total covers and total matchings~\cite{alavi1992,erdosmeir1977,meir1978,peledsun1994}. The problem can be viewed as a mixture of domination, edge domination, vertex cover, and edge cover.
	
	From the algorithmic point of view, the problem is hard on general graphs. Majumdar~\cite{majumdar1993} showed that mixed domination is NP-complete. It remains NP-complete on some restricted graph classes, including split graphs and planar bipartite graphs of maximum degree four~\cite{zhao2011,manlove1999}. These negative results motivate the study of graph classes where the problem becomes tractable.
	
	On the positive side, polynomial-time algorithms are known for several special graph classes. The unweighted problem has been solved on trees, cactus graphs, generalized series-parallel graphs, and proper interval graphs~\cite{lanchang2013, majumdar1993,madathil2019,rajaati2018}. Total matchings and total coverings have also been studied directly on threshold graphs~\cite{peledsun1994}. 
    More generally, by standard extensions of Courcelle's Theorem~\cite{courcelle1990} to MSO-definable optimization problems~\cite{arnborg1991easy}, one obtains polynomial-time algorithms for the MWMDS on graphs of bounded treewidth.

	There is also work on approximation, exact exponential-time, and parameterized algorithms. Hatami~\cite{hatami2007} gave a 2-approximation algorithm. Jain, Jayakrishnan, Panolan, and Sahu~\cite{jain2017} studied the parameterized complexity of the problem. Xiao and Sheng~\cite{xiaosheng2020} improved algorithms and kernels for the solution-size parameter. Dublois, Lampis, and Paschos~\cite{dublois2021} gave improved exact and parameterized algorithms, including results for treewidth and pathwidth.
	
	The weighted version is less developed in the literature. Xiao~\cite{xiao2019} studied approximation bounds for a restricted weighted model, where all vertices have the same weight and all edges have the same weight. Recently, Ferrarini, Kober, Lancini, and Yuditsky~\cite{ferrariniKoberLanciniYuditsky} proved that the minimum weight mixed dominating set problem is polynomially solvable for cliques, complete $k$-partite graphs, and blow-ups of constant size graphs. Moreover, they proved that the problem is NP-hard on cographs.
	
	The unweighted problem is the special case in which every weight equals one. Consequently, NP-hardness of the unweighted problem implies NP-hardness of the weighted problem, while a polynomial-time algorithm for the weighted problem also solves the unweighted case.
	
	\subsection{Contribution}
	In this work, we study the minimum weight mixed dominating set problem for threshold graphs. We prove that this problem can be solved in polynomial time in the size of the instance, and we provide an algorithm running in $\mathcal{O}(n^5)$. This result extends the previous results of Peled and Sun~\cite{peledsun1994} that gave a polynomial algorithm for the unweighted case. Since the MDS is NP-hard for split graphs~\cite{lanchang2013}, and the MWMDS is NP-hard for cographs~\cite{ferrariniKoberLanciniYuditsky} we have no hope to find a polynomial algorithm for these two  superclasses\footnote{Under the hypothesis that P$\neq$NP.}.
	Our result implies that the MWMDS is polynomially solvable for complete graphs and complete split graphs, results that were already given in~\cite{ferrariniKoberLanciniYuditsky}.
	\section{Preliminaries}
	Throughout the paper, we will only consider connected graphs with at least two vertices.
	
	Let $G = (V, E)$ be an undirected graph, and let $w$ be a non-negative weight function on both edges and vertices of $G$. A \emph{mixed dominating set} (MDS) is a set $D = D_V \cup D_E$, where $D_V \subseteq V$ and $D_E \subseteq E$, such that every vertex $v \in V$ and every edge $e \in E$ is dominated by at least one element in $D$.
	\begin{itemize}
		\item A vertex $v \in D_V$ dominates itself, all its adjacent vertices, and all its incident edges.
		\item An edge $e \in D_E$ dominates itself, its endpoints, and all edges sharing an endpoint with $e$.
	\end{itemize}
	We study the weighted version, that is the \emph{minimum weight mixed dominating set} (MWMDS) problem.
	The objective of the MWMDS problem is to minimize $w(D) = \sum_{v \in D_V} w(v) + \sum_{e \in D_E} w(e)$, where $D$ is a mixed dominating set of $G$.

	Throughout the paper, we will use two concepts strictly related to domination, the vertex cover and the edge cover.

	\begin{definition}[Vertex Cover]
		A \emph{vertex cover} (VC) of a graph is a set of its vertices that includes at least one endpoint of every edge.
	\end{definition}

	\begin{definition}[Edge Cover]
		An \emph{edge cover} (EC) of a graph is a set of its edges such that every vertex of the graph is an endpoint of at least one edge of the set. 
	\end{definition}
	
	The concept of MDS generalizes both edge and vertex cover, as shown in the following.
	
	\begin{lemma}[Edge and Vertex Covers are MDS]\label{lem:MDSareVCorEC}
		Let $G=(V,E)$ be a connected graph. Then a subset of $E$ is an MDS if and only if it is an edge cover. Similarly, a subset of $V$ is an MDS if and only if it is a vertex cover.
	\end{lemma}
	\begin{proof}
		If $D_E$ is a mixed dominating set containing only edges, every vertex must be incident to an edge of $D_E$, so $D_E$ is an edge cover. Conversely, if $D_E$ is an edge cover, every vertex is dominated by an incident selected edge, and every edge shares an endpoint with a selected edge; hence $D_E$ is a mixed dominating set.
		
		If $D_V$ is a mixed dominating set containing only vertices, every edge must have a selected endpoint, so $D_V$ is a vertex cover. Conversely, suppose that $D_V$ is a vertex cover. Every edge is dominated by a selected endpoint. Every vertex outside $D_V$ has a neighbor because $G$ is connected and has at least two vertices, and all of its neighbors must lie in $D_V$; hence every vertex is also dominated.
	\end{proof}

    \begin{definition}[Touched vertex]
        Let $S$ be a mixed dominating set of a graph $G$. We say that a vertex $v$ is \emph{touched} by $S$ if either $v\in S$ or $\delta(v)\cap S\neq \varnothing$. In other words, a vertex is touched if either itself or an edge incident to it is in this mixed set. 
    \end{definition}
    
	The set of touched vertices characterizes when a mixed set dominates all the edges of a graph. Indeed, we have the following.	
	\begin{lemma}[Edge Domination Equivalence] \label{lem:edge_dom}
		Let $G=(V,E)$ be a graph. A mixed set $S$ dominates all edges of $G$ if and only if the set of touched vertices of $G$ is a vertex cover. 
	\end{lemma}
	
	\begin{proof}
		Let $D=(D_V, D_E)$ be a mixed set and $X$ be the set of vertices touched by $D$.
		
		$\Leftarrow$) Suppose $X$ is a VC. For any edge $e$, at least one of its endpoints, say $u$, is in $X$. Thus, either $u \in D_V$ or $u \in V(D_E)$. In the first case, $u$ dominates $e$. In the second, there exists an edge, possibly $e$ itself, incident to $u$ that belongs to $D_E$. Such an edge dominates $e$. Hence, all edges are dominated.
		
		$\Rightarrow$) Suppose $X$ is not a VC. Let $e$ be an edge not covered by $X$. By definition of $X$, no endpoints of $e$ belong to $D_V$. Similarly, no endpoints of $e$ are endpoints of an edge in $D_E$. Therefore, $e$ is not dominated by either a vertex nor an edge, a contradiction.
	\end{proof}
	
	Throughout the paper, we will study the MWMDS problem for threshold graphs.
	
	\begin{definition}[Threshold Graph]
		A graph $G = (V, E)$ is a threshold graph if its vertex set $V$ can be partitioned into a clique $K=\{k_1, \dots, k_p\}$ and an independent set $I= \{i_1, \dots, i_q\}$ such that $N(i_1) \subseteq \dots \subseteq N(i_q)$ and $N[k_1]  \subseteq \dots \subseteq N[k_p]$. 
	\end{definition}
    
	We give an example of a threshold graph in Figure~\ref{fig:threshold_graph}.
	We use $N_I(x)=N(x)\cap I$ to denote the neighborhood of $x$ in $I$.
	
	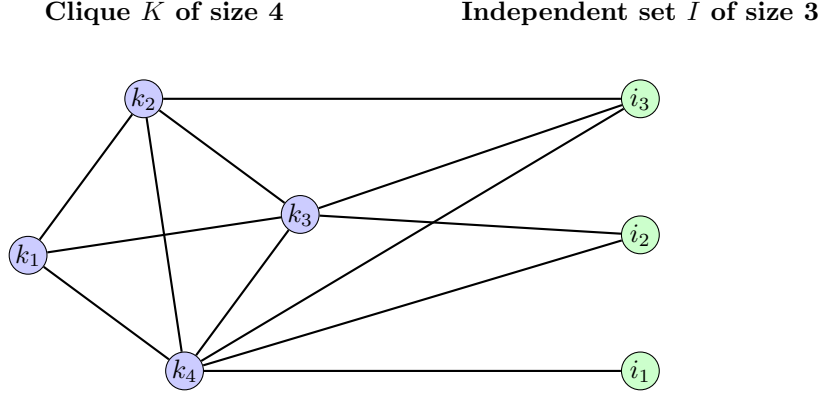
\begin{figure}[!ht]
		\centering
		\begin{tikzpicture}[scale=0.9, every node/.style={circle, draw, minimum size=0.5cm, inner sep=0pt}]
			\node[fill=blue!20] (k1) at (-2, 1.7) {$k_1$};
			\node[fill=blue!20] (k2) at (-.3, 4) {$k_2$};
			\node[fill=blue!20] (k4) at (.3, 0) {$k_4$};
			\node[fill=blue!20] (k3) at (2, 2.3) {$k_3$};
			
			\node[fill=green!20] (i1) at (7, 4) {$i_3$};
			\node[fill=green!20] (i2) at (7, 2) {$i_2$};
			\node[fill=green!20] (i3) at (7, 0){$i_1$};
			
			\draw[thick] (k1) -- (k2); 
			\draw[thick] (k1) -- (k3); 
			\draw[thick] (k1) -- (k4);
			\draw[thick] (k2) -- (k3); 
			\draw[thick] (k2) -- (k4); 
			\draw[thick] (k3) -- (k4);
			
			\draw[thick] (i1) -- (k4); 
			\draw[thick] (i1) -- (k2); 
			\draw[thick] (i1) -- (k3);
			\draw[thick] (i2) -- (k4); 
			\draw[thick] (i2) -- (k3); 
			\draw[thick] (i3) -- (k4);
			
			\node[shape=rectangle, draw=none] at (0, 5.25) {\textbf{Clique $K$ of size 4}};
			\node[shape=rectangle, draw=none] at (7, 5.25) {\textbf{Independent set $I$ of size 3}};
		\end{tikzpicture}
		\caption{A threshold graph}
		\label{fig:threshold_graph}
	\end{figure}
	
	For threshold graphs we are able to characterize all inclusion-wise minimal vertex covers. This will prove useful for the proof of correctness of our algorithm.
	\begin{lemma}[Minimal Vertex Covers of Threshold Graph] \label{lem:min_vc_tg}
		All inclusion-wise minimal vertex covers of a threshold graph $G=(K \cup I, E)$ belong to the collection $\mathcal{V}(G) = \{K\} \cup \{N(x) \mid x \in K \}$.
	\end{lemma}
	\begin{proof}
		Let $C$ be a vertex cover of $G$. Since $K$ is a clique, $C$ must contain at least $|K|-1$ vertices of $K$ to cover $E(K)$.
		
		If there exists $x \in K$ such that $x\notin C$, then all edges of $\delta(x)$ must be covered by vertices in $N(x)$. In particular, all the vertices of $N_I(x)$ belong to $C$. Therefore, $N(x)$ is a minimal vertex cover.

		If $K\subseteq C$, since $K$ covers all internal edges and all cross edges, $K$ itself is a vertex cover. If all vertices of $K$ are adjacent to at least one vertex of $I$, then this vertex cover is minimal.
        Otherwise, for each $k\in K$ that is adjacent only to vertices of $K$, the set $K\setminus\{k\}=N(k)$ is a minimal vertex cover.
        In every case, these sets belong to the collection $\mathcal{V}(G)$.
	\end{proof}
	
	\subsection{Arbitrary Weights}
	
	In the following section, we suppose that the weight function $w$ takes nonnegative values. This hypothesis allows us to heavily simplify the proof of our results. However, in this section, we prove that this hypothesis is nonrestrictive.

	Given an instance of the MWMDS for a graph $G=(V,E)$, the weight of every element $x \in V \cup E$ can be decomposed as $w(x) = w^+(x) + w^-(x)$, where $w^+(x) = \max(w(x), 0)$ and $w^-(x) = \min(w(x), 0)$. 
	
	Let $V^- = \{v \in V \mid w(v) < 0\}$ and $E^- = \{e \in E \mid w(e) < 0\}$ denote the sets of strictly negative-weight vertices and edges respectively. Let $W^-$ be the sum of all negative weights in the graph, that is:
	$$W^- = \sum_{v \in V^-} w(v) + \sum_{e \in E^-} w(e).$$

	\begin{lemma} \label{lem:weight_lower_bound}
		For any mixed dominating set $D = D_V \cup D_E$ of graph $G$, its total weight under $w$ satisfies $w(D) \ge w^+(D) + W^-$.
	\end{lemma}
	\begin{proof}
		By definition of our problem we have
		$$w(D) = \sum_{x \in D} w(x) = \sum_{x \in D} w^+(x) + \sum_{x \in D} w^-(x) = w^+(D) + w^-(D).$$
		Since $D \cap (V^- \cup E^-) \subseteq V^- \cup E^-$, we have that $w^-(D) \ge W^-$ which yields $w(D) \ge w^+(D) + W^-$.
	\end{proof}
	
	\begin{theorem} \label{thm:arbitrary_reduction}
		Let $G=(V,E)$ be a graph and $w$ an arbitrary weight function on the elements of $G$. Let $D^+$ be the optimum of the MWMDS problem on $G$ with weights $w^+$.
		Then, the optimal solution of the MWMDS problem on $G$ with weights $w$ is given by $D=D^+\cup V^-\cup E^-$.
	\end{theorem}
	\begin{proof}
		Since the set $D^+$ is a mixed dominating set for $G$ and the mixed domination property is preserved under element addition, $D$ is a mixed dominating set.
		
		The weight of $D$ is $w(D) = w^+(D) + w^-(D)$. By construction, $D$ contains all negative elements $V^-$ and $E^-$, hence $w^-(D) = W^-$. Additionally, for all $x \in V^- \cup E^-$, the non-negative weight is $w^+(x) = 0$. Adding these elements to $D^+$ does not change the non-negative weight sum, meaning $w^+(D) = w^+(D^+)$. Consequently, the exact weight of $D$ is
		$ w(D) = w^+(D^+) + W^-$.
		
		Now, let $S$ be any mixed dominating set of $G$. Because $D^+$ is the solution minimizing the $w^+$ objective, it must hold that $w^+(D^+) \le w^+(S)$. Applying Lemma~\ref{lem:weight_lower_bound} to $S$, we obtain
		$$ w(D) = w^+(D^+) + W^- \le w^+(S) + W^- \le w(S). $$
		
		Therefore, $D$ is a mixed dominating set of minimum weight for $G$ with weights $w$.
	\end{proof}
	The pre-processing of weights and post-processing union operate in linear time. Therefore, the overall time complexity of the algorithm remains unchanged after these operations. Henceforth, we may assume that all weights are nonnegative.

	\section{Results}

    \subsection{Constrained Mixed Cover Algorithm}
    The Constrained Mixed Cover formulation and its reduction to minimum-weight edge cover were first developed in the ongoing work of Ferrarini, Kober, Lancini, and Yuditsky~\cite{ferrariniKoberLanciniYuditsky}. We recall the formulation and reduction here, with minor modifications, because they form a building block of our threshold-graph algorithm.
    This problem is a variant of the MWMDS where we partition the set of vertices into four subsets:
    \begin{itemize}
        \item a set $S$ of \emph{solution vertices}, that is,  vertices that must belong to the solution,
        \item a set $T$ of \emph{touched vertices}, that is, vertices that must be touched by the solution, and can be part of the solution,
        \item a set $B$ of \emph{banned vertices}, that is, vertices that must be touched by the solution, but cannot belong to the solution,
        \item a set $F$ of \emph{free vertices}, that is, vertices that cannot belong to the solution, and that are not required to be touched by the solution.
    \end{itemize}
    
    \begin{definition}[Constrained Mixed Cover problem]\label{def:CMCP}
        Let $G=(V,E)$ be a graph, let $w$ be a weight function on both vertices and edges of $G$, and let $V$ be partitioned into four sets of vertices $S$, $T$, $B$, and $F$ mutually disjoint. The \emph{Constrained Mixed Cover problem} $\mathcal{H}(S, T, B, F)$ consists of finding a mixed set $(D_V \subseteq V, D_E \subseteq E)$ minimizing $w(D_V) + w(D_E)$ such that: $S\subseteq D_V$, $T \subseteq D_V \cup V(D_E)$, $D_V \cap (B \cup F) = \varnothing$, and $B\subseteq V(D_E)$.

    \end{definition}
	One possible way to solve the MWMDS for a graph $G=(V,E)$ consists of solving a (possibly exponentially large) sequence of constrained mixed cover problems: if we solve $\mathcal{H}(S, T, B, F)$ for all possible partitions of $V$ such that $F$ is a stable set and each element of $F$ is adjacent to an element of $S$, we completely explore the space of solutions to the MWMDS.

    Clearly, this strategy makes sense only under two specific conditions: $\mathcal{H}(S, T, B, F)$ is solvable in polynomial time, and we are able to reduce our search to a polynomial number of subproblems. Fortunately, the first condition holds thanks to a reduction to the edge cover problem, as we proceed to show.
 
	\paragraph{Construction.}

    We now describe a reduction from $\mathcal{H}(S,T,B,F)$ to an instance of the minimum-weight edge-cover problem. To this end, we construct an auxiliary weighted graph $ \bar G=(\bar V,\bar E,\bar w)$.
    This construction first appeared, with some minor changes, in Ferrarini, Kober, Lancini, and Yuditsky~\cite{ferrariniKoberLanciniYuditsky}. 
The construction consists of a copy of the original graph, together with auxiliary vertices and edges used to encode the relevant vertex weights. The vertex set $\bar{V}$ is given by:
$$\bar{V}= \{\bar v : \forall v\in V\} \cup \{v' : \forall v\in V\setminus B\}\cup \{v'' : \forall v\in T\}.$$
Thus, for every vertex $v\in V$, the graph $\bar G$ contains a corresponding vertex $\bar v$. Moreover, for every $v\in V\setminus B$ we introduce an auxiliary vertex $v'$. A second auxiliary vertex $v''$ is introduced only for vertices $v\in T$.

On the other hand, the edge set is defined as follows:
$$ \bar E= \{\bar u\bar v : \forall uv\in E\}\cup\{\bar v v' : \forall v\in V\setminus B\} \cup \{v'v'' : \forall v\in T\}.$$
In particular, every edge $uv\in E$ corresponds to an edge
$\bar u\bar v$. Moreover, each auxiliary vertex $v'$ is adjacent to $\bar v$, and each $v''$ is adjacent only to $v'$.

It remains to define the edge-weight function $\bar{w}$.
For every edge $uv\in E$, we set $\bar{w}(\bar{u}\bar{v})=w(uv)$.
For every $v\in V\setminus B$, we set:
\[ \bar{w}(\bar{v}v') = \begin{cases} w(v)\quad &\text{if } v\in S\cup T,\\
0, & \text{otherwise}. \end{cases} \] 
Finally, for every $v\in  T$, we set $\bar{w}(v'v'')=0$.
This completes the construction of the auxiliary weighted graph $\bar{G}$. It is important to notice that the reduction is polynomial.
	
	\begin{example}\label{ex:auxiliary_graph_weights}
		Consider the graph $G$ given in Figure~\ref{fig:auxiliary_graph}. Let $S=\{v_1\}$, $T=\{v_3\}$, and $B=\{v_2\}$, and $F=\{v_4, v_5\}$. The number next to each vertex and edge is its weight. On the right we have the auxiliary graph $\bar{G}$, whose edges are labeled by their weight $\bar w$. Thus, the edges $\bar{v_i}v'_i$ for $i\in\{1,3,4,5\}$ have weights $2,3,0,0$ respectively, while every edge $v'_iv''_i$ has weight $0$.
		\begin{figure}[ht]
			\begin{center}
				\begin{tikzpicture}[
					scale=1.4,
					vertex/.style={circle, draw, minimum size=0.7cm, inner sep=0pt},
					weight/.style={rectangle, draw=none, fill=none, inner sep=1pt, font=\small},
					graphlabel/.style={rectangle, draw=none, fill=none, inner sep=0pt}
					]
					\node[vertex, fill=orange!20] (g1) at (-4.0, 1.2) {$v_1$};
					\node[vertex, fill=red!30](g2) at (-5.1, 0.4) {$v_2$};
					\node[vertex, fill=teal!20](g3) at (-4.7, -1.0) {$v_3$};
					\node[vertex, fill=blue!20](g4) at (-3.3, -1.0) {$v_4$};
					\node[vertex, fill=blue!20](g5) at (-2.9, 0.4) {$v_5$};
					\draw (g1) -- node[weight, above left, xshift=-1pt, yshift=1pt] {$4$} (g2);
					\draw (g1) -- node[weight, right, xshift=2pt] {$3$} (g4);
					\draw (g1) -- node[weight, above, yshift=1pt] {$6$} (g5);
					\draw (g2) -- node[weight, below, yshift=-1pt] {$2$} (g4);
					\draw (g2) -- node[weight, above, yshift=1pt] {$5$} (g5);
					\draw (g4) -- node[weight, right, xshift=1.5pt] {$1$} (g5);
					\draw (g2) -- node[weight, left, xshift=-1.5pt] {$7$} (g3);
					\draw (g4) -- node[weight, below, yshift=-1pt] {$4$} (g3);
					\node[weight, above=1pt of g1] {$2$};
					\node[weight, left=1pt of g2] {$5$};
					\node[weight, below left=1pt and 0pt of g3] {$3$};
					\node[weight, below right=1pt and 0pt of g4] {$4$};
					\node[weight, right=1pt of g5] {$6$};
					\node[graphlabel] at (-4.0, -1.85) {$G$};
					
					\node[vertex, fill=orange!20](v1) at (2.0, 1.2) {$\bar{v_1}$};
					\node[vertex, fill=red!30](v2) at (0.9, 0.4) {$\bar{v_2}$};
					\node[vertex, fill=teal!20](v3) at (1.3, -1.0) {$\bar{v_3}$};
					\node[vertex, fill=blue!20](v4) at (2.7, -1.0) {$\bar{v_4}$};
					\node[vertex, fill=blue!20](v5) at (3.1, 0.4) {$\bar{v_5}$};
					\draw (v1) -- node[weight, above left, xshift=-1pt, yshift=1pt] {$4$} (v2);
					\draw (v1) -- node[weight, right, xshift=2pt] {$3$} (v4);
					\draw (v1) -- node[weight, above, yshift=1pt] {$6$} (v5);
					\draw (v2) -- node[weight, below, yshift=-1pt] {$2$} (v4);
					\draw (v2) -- node[weight, above, yshift=1pt] {$5$} (v5);
					\draw (v4) -- node[weight, right, xshift=1.5pt] {$1$} (v5);
					\draw (v2) -- node[weight, left, xshift=-1.5pt] {$7$} (v3);
					\draw (v4) -- node[weight, below, yshift=-1pt] {$4$} (v3);
					
					\node[vertex](v1a) at (2.00, 2.10) {$v_1'$};
					\node[vertex](v3a) at (0.35, -1.13) {$v_3'$};
					\node[vertex](v3b) at (0.05, -0.30) {$v_3''$};
					\node[vertex](v4a) at (3.65, -1.13) {$v_4'$};
					\node[vertex](v5a) at (4.00, 0.70) {$v_5'$};
					\draw (v1) -- node[weight, text=black, right, xshift=1pt] {$2$} (v1a);
					\draw (v3) -- node[weight, text=black, below, yshift=-1pt] {$3$} (v3a);
					\draw (v3a) -- node[weight, text=black, left, xshift=-1pt] {$0$} (v3b);
					\draw (v4) -- node[weight, text=black, below, yshift=-1pt] {$0$} (v4a);
					\draw (v5) -- node[weight, text=black, above, yshift=1pt] {$0$} (v5a);
					\node[graphlabel] at (2.0, -1.85) {$\bar{G}$};
				\end{tikzpicture}
				\caption{The weighted graph $G$ and its auxiliary graph $\bar{G}$.}
				\label{fig:auxiliary_graph}
			\end{center}
		\end{figure}
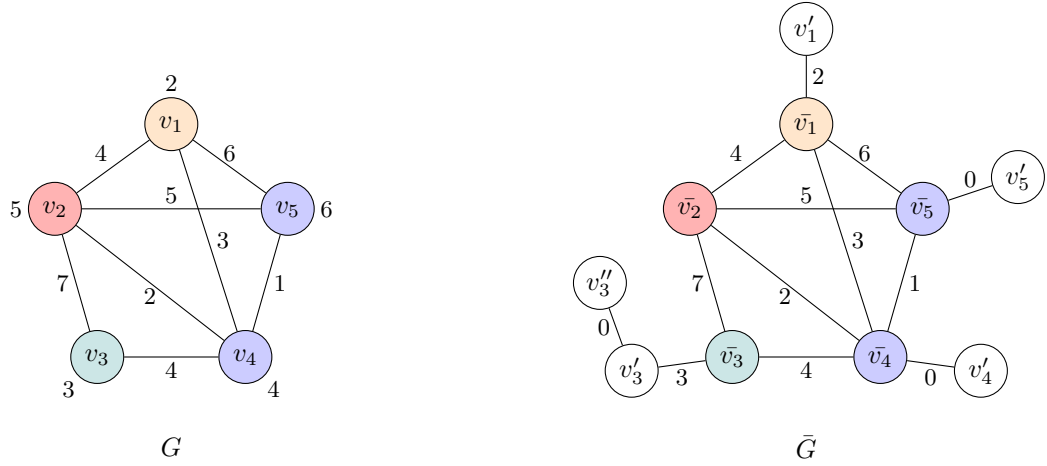
	\end{example}
	
	\begin{proposition} \label{prop:algorithm}
		Given a graph $G=(V,E)$, for any $S, T, B, F$ as in Definition~\ref{def:CMCP}, $\mathcal{H}(S, T, B, F)$ is equivalent to solving the Minimum Weight Edge Cover problem on $\bar{G}$. 
	\end{proposition}
	
	\begin{proof}
		The proof is divided into two parts. We first prove that for any $D$ solution to $\mathcal{H}(S, T, B, F)$ there exists an edge cover of $\bar{G}$ of the same total weight. Then, we show that for every edge cover of $\bar{G}$, there exists a mixed set $D$ of same total weight that respects the conditions on the vertex partition.
		
		Let $D=(D_V, D_E)$ be a mixed set that touches $T$, contains $S$, and has no vertex in $B\cup F$.
		Let $E_0$ be the set of edges of $\bar{G}$ defined as:
        $$E_0=\{\bar{e} \in \bar{E} : e\in D_E\}.$$
        Let $E_1\subseteq\bar{E}$ be defined as:
		$$\{\bar{v}v'  \in \bar{E}: v\in D_V\}\cup\{\bar{e}\in \bar{E} :\bar{w}(\bar{e})=0\}.
		$$
		We claim that $\bar D=E_0\cup E_1$ is an edge cover of $\bar{G}$ such that $w(D)=\bar{w}(\bar{D})$.
		The equivalence of weights stems from the fact that the edges of $G$ have the same weights of their counterparts in $\bar{G}$, and the fact that $\bar{w}(\bar{v}v')=w(v)$ for all $v\in S\cup T$.
		We claim that the set of edges $\bar D$ is an edge cover of $\bar G$. 
        For every $v\in D_V$, the edge $\bar{v}v'$ belongs to $\bar{D}$, in particular, for all $v\in S$, the vertices $\bar{v}$ and $v'$ are covered by $\bar{v}v'$. For every $v\in (T\cup B)\setminus D_V$ there exists an edge $e$ in $\delta(v)\cap D$, and the corresponding edge $\bar{e}$ belongs to $\bar{D}$, therefore for all $v\in S\cup T\cup B$, $\bar{v}$ is covered. For all $v\in T$, the vertices $v'$ and $v''$ of $\bar{G}$ are covered by $v'v''$ since this edge has weight $0$, similarly, for all $v\in F$ the vertex $v'$ is covered by the edge $\bar{v}v'$ that has weight 0. Therefore, $\bar{D}$ is an edge cover.
		
		We now prove that, given an edge cover $\bar{D}$ of $\bar{G}$, there exists a mixed set $D\subseteq V\cup E$ of the same weight that is solution to $\mathcal{H}(S,T,B,F)$. 
        
		Let $V_D=\{v\in S\cup T : \bar{v}v'\in \bar{D}\}$, and $E_D=\{e\in E: \bar{e}\in \bar{D} \}$. The set $D^*=(V_D, E_D)$ is solution to $\mathcal{H}(S,T,B,F)$ of the same weight as $\bar{D}$.
        By construction, for all vertices $v\in S$, the edge $\bar{v}v'$ belongs to $\bar{D}$.
        For each $v\in T$, either $\bar{v}v'$ belongs to $\bar{D}$ or there exists a vertex $\bar{u}$ such that $\bar{u}\bar{v}$ belongs to $\bar{D}$. In the first case $v\in V_D$, in the second, $v$ is touched by the edge $uv\in E_D$.
        Furthermore, for all $v\in B$, there exists a vertex $\bar{u}$ such that $\bar{u}\bar{v}$ belongs to $\bar{D}$, the corresponding edge belongs to $D^*$, and therefore $D^*$ touches all vertices of $B$ via an edge. By construction the vertices of $B$ and $F$ do not belong to $V_D$.		
		Since for all $e\in E$, we have $w(e)=\bar{w}(\bar{e})$, and for all $v\in S\cup T$ we have $w(v)=\bar{w}(\bar{v}v')$, we can conclude that $w(D^*)=\bar{w}(\bar{D})$.    \end{proof}
	\begin{corollary}\label{cor:HO3}
		For a graph $G=(V,E)$, solving the constrained mixed cover problem can be done in $\mathcal{O}(n^3)$ time.
	\end{corollary}
	\begin{proof}
		The minimum weight edge cover is solvable in $\mathcal{O}(n^3)$ for a general graph~\cite{schrijver2003combinatorial}. The transformation from $G$ to $\bar{G}$ adds at most $2n$ vertices and edges to the graph, thus it is polynomial. Consequently, we can solve $\mathcal{H}(S,T,B,F)$ in $\mathcal{O}(n^3)$ time.
	\end{proof}
	

	\subsection{Polynomial-Time Algorithm for Threshold Graphs}
	
	\begin{theorem} \label{thm:main}
		The MWMDS problem on a threshold graph $G=(V,E)$ can be solved in $\mathcal{O}(n^5)$ time.
	\end{theorem}
	\begin{proof}
		Let $G=(K\cup I, E)$ be a threshold graph, and let $w$ be a non-negative weight function.
		We distinguish three cases based on $D_V$.
		
		\paragraph{Case 1: $D_V = \varnothing$.} 
		Suppose that the mixed dominating set does not contain vertices, then we solve the following problem.
        \begin{enumerate}
            \item Solve $\mathcal{H}(\varnothing, \varnothing, V, \varnothing)$.
        \end{enumerate}
		
		\paragraph{Case 2: $D_V \subseteq I$ and $D_V \neq \varnothing$.} 
		Suppose that $D_V$ is not empty and contains only vertices in $I$.
        For each $i \in I$, we denote by $I_i^-$ the set of vertices $\{j\in I : N(j)\subseteq N(i), j\neq i\}$. Similarly, we denote by $I_i^+$ the set $I\setminus (I_i^-\cup \{i\})$.
		We solve the following sequence of constrained mixed cover problems.
		\begin{enumerate}[resume]
			\item For each $i$ in $I$, solve $\mathcal{H}(\{i\}, I\setminus \{i\}, K, \varnothing )$.
			\item For each $i\in I$ and each $k'\in N(i)$, solve $\mathcal{H}(\{i\}, I_i^-, I_i^+\cup K\setminus \{k'\}, \{k'\})$.
		\end{enumerate}
        \paragraph{Case 3: $D_V \cap K\neq \varnothing$.} 
        For every $k\in K$ define the set $B_k$ as
		$$B_k=\{v\in K : N[k]\subsetneq N[v]\}.$$
        We solve the following sequence of constrained mixed cover problems.
		\begin{enumerate}[resume]
			\item For each $k \in K$, solve $\mathcal{H}(\{k\}, V\setminus \{k\}, \varnothing, \varnothing )$.
            \item For each $k \in K$, solve  $\mathcal{H}(\{k\}, V\setminus (B_k\cup \{k\}\cup N_I(k)), B_k, N_I(k) )$.
			\item For each $k\in K$, for each $k'\in B_k$ solve $\mathcal{H}(\{k\}, V\setminus (B_k \cup \{k\}), B_k\setminus \{k'\},  \{k'\})$.
            \item For each $k\in K$, for each $k'\in K\setminus(B_k\cup\{k\})$, solve $\mathcal{H}(\{k\}, V\setminus(B_k\cup \{k, k'\} \cup (N_I(k)\setminus N_I(k')))  , B_k ,  \{k'\}\cup(N_I(k)\setminus N_I(k')))$.
		\end{enumerate}

        We claim that these subproblems all give a mixed dominating set as an optimal solution, and that every inclusion-wise minimal mixed dominating set is solution to at least one of these problems. 

    \begin{claim}
        Every feasible solution to each of the problems listed in items 1.-7.~is a mixed dominating set.
    \end{claim}
    \begin{proof}
        Problems in classes 1, 2 and 4 require all vertices of the graph to be touched by the mixed set, therefore every solution to these problems is trivially a mixed dominating set.

        Problems in classes 3 and 6 require all but one vertex to be touched. In particular, the only possibly untouched vertex $k'$ belongs to the clique and has a neighbor in the solution. This has two consequences: first, all vertices are either touched by the solution or they are dominated by a vertex, and second, the set of touched vertices contains $N(k')$, that is a vertex cover by Lemma~\ref{lem:min_vc_tg}. Thus, all edges are dominated by Lemma~\ref{lem:edge_dom}, and so the solutions are all mixed dominating sets.

        Problems in class 5 follow a similar principle. The set of possibly untouched vertices for this problem is $N_I(k)$ for some $k\in K$ that belongs to the solution.  Therefore, all vertices are either touched or dominated by $k$. In particular, all vertices of $K$ are touched, so all the edges are dominated by Lemmas~\ref{lem:edge_dom} and~\ref{lem:min_vc_tg}.

        For problems in class 7, the set of possibly untouched vertices is  $\{k'\}\cup(N_I(k)\setminus N_I(k'))\subseteq N(k)$ for $k,k' \in K$, where $k$ belongs to the solution. In this case, all vertices are either touched by the solution or dominated by $k$. 
        Moreover, all vertices of $N(k')$ are touched, therefore the set of touched vertices contains a vertex cover by Lemma~\ref{lem:min_vc_tg}, and thus all edges are dominated by Lemma~\ref{lem:edge_dom}. \end{proof}
        
    \begin{claim}
        Every inclusion-wise minimal mixed dominating set on $G$ is a feasible solution to at least one of the problems listed in items 1.-7. 
    \end{claim}      
\begin{proof}
    Let $D=(D_V,D_E)$ be an inclusion-wise minimal mixed dominating set of $G$.
    We follow the three cases of the proof.
    
    \paragraph{Case 1} If $D_V=\varnothing$, $D$ is an edge cover, and thus it is solution to $\mathcal{H}(\varnothing, \varnothing, V, \varnothing)$.
    
    \paragraph{Case 2} Suppose that $\varnothing\neq D_V\subseteq I$, then if all vertices of $G$ are touched by $D$, and $i\in I\cap D_V$, then $D$ is solution to $\mathcal{H}(\{i\}, I\setminus \{i\}, K, \varnothing )$.
    On the contrary,  if there are untouched vertices they all belong to $K$, because the vertices of $I$ can be dominated only by vertices in $K$ (that do not belong to the solution by hypothesis). Since two untouched vertices cannot be adjacent, we have that there exists at most one of them, say $k'$, that must be adjacent to a vertex of $I\cap D_V$, as otherwise it is not dominated. Let $i$ be the vertex of $I\cap D_V$ with the maximum neighborhood, that is, no vertex of $I^+_i$ belongs to $D$. Therefore $D$ is solution to $\mathcal{H}(\{i\}, I_i^-, I_i^+\cup K\setminus \{k'\}, \{k'\})$. 

    \paragraph{Case 3} Suppose that $D_V\cap K\neq \varnothing$.
    We denote by $k$ the vertex of $D_V\cap K$ of maximum neighborhood.
    
    If all vertices of $G$ are touched by $D$, then $D$ is solution to $\mathcal{H}(\{k\}, V\setminus \{k\}, \varnothing, \varnothing )$.
    
    If there exist untouched vertices, they must form a stable set. Therefore, at most one vertex of $K$ is untouched.

    If all vertices of $K$ are touched, then the set of untouched vertices is a subset of $N_I(k)$. Moreover, since $D$ is inclusion-wise minimal by hypothesis, no vertices of $N_I(k)$ belong to $D_V$: they and their neighbors are all dominated by $k$, similarly all edges between them and $K$ are dominated by whatever touches the other endpoint. This implies that, in this subcase, $D$ is solution to  $\mathcal{H}(\{k\}, V\setminus (B_k\cup \{k\}\cup N_I(k)), B_k, N_I(k) )$.

    Therefore, suppose that there exists a vertex $k'\in K$ not touched by $D$. Since $G$ is a threshold graph, either $N[k]\subsetneq N[k']$ or $N[k']\subseteq N[k]$.
    In the first case, that is $k'\in B_k$, we have that the vertices of $I$ are either adjacent to $k'$ or not dominated by $k$. Since $k$ is the vertex of $D_V\cap K$ of maximum neighborhood, this implies that the vertices not dominated by $k$ are not dominated by any vertex, and thus they must be touched. On the other hand, if they are adjacent to $k'$, they must be touched since the set of untouched vertices is a stable set. Consequently, in this case $D$ is solution to $\mathcal{H}(\{k\}, V\setminus (B_k \cup \{k\}), B_k\setminus \{k'\},  \{k'\})$.
    In the second case, we have that the set of vertices that are adjacent to $k$ and not adjacent to $k'$ are not required to be touched, as they are dominated by $k$ and all their neighbors are touched. Moreover, since $D$ is inclusion-wise minimal by hypothesis, we can deduce that these vertices do not belong to $D$. In this case $D$ is solution to $\mathcal{H}(\{k\}, V\setminus(B_k\cup \{k, k'\} \cup (N_I(k)\setminus N_I(k')))  , B_k ,  \{k'\}\cup(N_I(k)\setminus N_I(k')))$.

    We extensively treated all possible cases, each case yields a solution to one of the problems above, thus we proved the claim.\end{proof}

    Since $w$ is a nonnegative function, there exists an optimal solution to the $MWMDS$ problem on $G$ that is inclusion-wise minimal.
    Therefore, once we solve all possible problems of the list above, we obtain a set of possible optimal solutions. 
    Taking the best is sufficient to obtain the optimum.

    The proof is concluded once we realize that each item of the list above consists in solving at most $|V|^2$ constrained mixed cover problems. Each of those is solvable in $\mathcal{O}(|V|^3)$ time by Corollary~\ref{cor:HO3}. Thus, the overall complexity of the problem is $\mathcal{O}(|V|^5)$.\end{proof}
	
	\bibliographystyle{plain}
	\bibliography{references}
\end{document}